\documentclass{LMCS}

\def\dOi{10(1:6)2014}
\lmcsheading%
{\dOi}
{1--15}
{}
{}
{Jan.~\phantom01, 2013}
{Feb.~11, 2014}
{}

\ACMCCS{[{\bf Theory of computation}]: Computational complexity and
  cryptography---Problems, reductions and completeness; [{\bf
      Mathematics of computing}]: Mathematical analysis---Numerical
  analysis---Interval arithmetic; Mathematical analysis---Differential
  equations---Ordinary differential equations} 


\pdfoutput=1

\usepackage{amsmath,amssymb,hyperref}
\usepackage{graphicx}

\usepackage{amsthm}
\newtheorem{theorem}{Theorem}
\newtheorem{lemma}[theorem]{Lemma}

\theoremstyle{definition}
\newtheorem{definition}[theorem]{Definition}
\usepackage{CJK,CJKspace,CJKpunct}
\pdfmapline{=unisong@Unicode@ <ipam.ttf}

\newcommand{\R}{\mathbf R}
\newcommand{\N}{\mathbf N}
\newcommand{\Q}{\mathbf Q}
\newcommand{\Z}{\mathbf Z}
\newcommand{\D}{D}
\newcommand{\classP}{\mathsf{P}}
\newcommand{\classPSPACE}{\mathsf{PSPACE}}
\newcommand{\classNP}{\mathsf{NP}}
\newcommand{\classNumberP}{\mathsf{\#P}}
\newcommand{\classPH}{\mathsf{PH}}
\newcommand{\classPP}{\mathsf{PP}}
\newcommand{\classCH}{\mathsf{CH}}
\newcommand{\classSigma}{\mathsf{\Sigma}}
\newcommand{\quantC}{\mathsf{C}}

\newcommand{\OpIVP}{\mathit{ODE}}
\newcommand{\deltabox}{\delta _\square}
\newcommand{\deltaboxLip}{\delta _{\square \mathrm L}}
\newcommand{\deltaTaylor}{\delta _{\mathrm{Taylor}}}
\newcommand{\classtwofont}[1]{\text{\bfseries \sffamily \upshape #1}}
\newcommand{\classFPtwo}{\classtwofont{FP}}
\newcommand{\classFPSPACEtwo}{\classtwofont{FPSPACE}}

\newcommand{\redW}{\leq _{\mathrm W}}
\newcommand{\classLip}{\mathrm{CL}}
\newcommand{\classC}{\mathrm C}

\begin{document}

\title[Computational Complexity of Smooth Differential Equations]{Computational Complexity of\\Smooth Differential Equations\rsuper*}
\titlecomment{%
{\lsuper*}Some of the results in this paper have been reported at conferences
\cite{eatcs-la, mfcs}. 
This work was supported in part by: 
\emph{Kakenhi} (Grants-in-Aid for Scientific Research, Japan) 23700009 and 24106002; 
\emph{Marie Curie International Research Staff Exchange Scheme Fellowship}~294962, 
the 7th European Community Framework Programme; and 
\emph{German Research Foundation} with project DFG Zi\,1009/4-1.}

\author[A.~Kawamura]{Akitoshi Kawamura}  
\author[H.~Ota]{Hiroyuki Ota}            
\author[C.~R\"osnick]{Carsten R\"osnick} 
\author[M.~Ziegler]{Martin Ziegler}      

\begin{abstract}
The computational complexity of the solution~$h$ to 
the ordinary differential equation 
$h(0)=0$, $h'(t) = g(t, h(t))$ 
under various assumptions on the function $g$
has been investigated. 
Kawamura showed in 2010 that the solution~$h$ can be $\classPSPACE$-hard
even if $g$ is assumed to be Lipschitz continuous and polynomial-time computable. 
We place further requirements on the smoothness of $g$ 
and obtain the following results: 
the solution~$h$ can still be $\classPSPACE$-hard
if $g$ is assumed to be of class $\classC ^1$; 
for each $k \geq 2$, 
the solution~$h$ can be hard for the counting hierarchy 
even if $g$ is of class $\classC ^k$. 
\end{abstract}

\maketitle

\section{Introduction}

Let $g \colon [0,1] \times \R \to \R$ be continuous 
and consider the differential equation 
\begin{align}
 \label{eq:ode}
 h(0) & = 0, &
 \D h(t) & = g(t,h(t)) \quad t \in [0,1], 
\end{align}
where $\D h$ denotes the derivative of $h$. 
How complex can the solution~$h$ be, 
assuming that $g$ is polynomial-time computable? 
Here, polynomial-time computability 
and other notions of complexity 
are from the field of 
\emph{Computable Analysis} 
\cite{ko1991complexity,weihrauch00:_comput_analy}
and measure how hard it is to 
approximate real functions with specified precision 
(Section~\ref{section: preliminaries}). 

If we make no assumption on $g$ other than being polynomial-time computable, 
the solution~$h$ (which is not unique in general) can be non-computable. 
Table~\ref{table:related} summarizes known results about 
the complexity of $h$ under various assumptions 
(that get stronger as we go down the table). 
In particular, 
as the third row says, 
if $g$ is (globally) Lipschitz continuous, 
then the (unique) solution $h$ is known to be 
polynomial-space computable but still can be 
$\classPSPACE$-hard \cite{kawamura2010lipschitz}. 
In this paper, we study the complexity of $h$ 
when we put stronger assumptions about 
the smoothness of $g$. 

\begin{table}
\renewcommand\arraystretch{1.3}
\begin{center}
 \caption{Complexity of the solution $h$ of \eqref{eq:ode}
 assuming $g$ is polynomial-time computable}
 \label{table:related}
\small\vspace*{4pt}
 \begin{tabular}{lll}
  Assumptions & Upper bounds & Lower bounds \\
  \noalign{\smallskip}
  \hline
  \noalign{\smallskip}
   --- & --- & can be all non-computable \cite{pour1979computable} \\
  $h$ is the unique solution & computable \cite{coddington1955theory}
  & \parbox[t]{14em}{can take arbitrarily long time\\\cite{ko1983computational,miller1970recursive}} \\
  the Lipschitz condition  & polynomial-space \cite{ko1983computational}
      &	can be $\classPSPACE$-hard \cite{kawamura2010lipschitz}\\
  $g$ is of class $\classC ^{(\infty, 1)}$ & polynomial-space 
      & \parbox[t]{14em}{can be $\classPSPACE$-hard\\(Theorem~\ref{DifferentiableIsPspace})} \\
  \parbox[t]{10.55em}{$g$ is of class $\classC ^{(\infty, k)}$\\{}(for each constant $k$)}
  & polynomial-space 
  & can be $\classCH$-hard (Theorem~\ref{KTimesIsCH}) \\
  $g$ is analytic
  & polynomial-time \cite{muller1987uniform,ko1988computing,kawamura2010complexity} 
  & ---
 \end{tabular}
\end{center}
\end{table}

In numerical analysis, 
knowledge about smoothness of the input function 
(such as being differentiable enough times) 
often helps
to apply certain algorithms or simplify their analysis.
However, 
to our knowledge, 
this casual understanding that smoothness is good 
has not been rigorously substantiated 
in terms of computational complexity theory. 
This motivates us to ask whether, 
for our differential equation \eqref{eq:ode}, 
smoothness really reduces the complexity of the solution. 

At one extreme is the case where $g$ is analytic: 
$h$ is then polynomial-time computable 
(the last row of the table) 
by an argument based on Taylor series\footnote{
As shown by M\"uller \cite{muller1987uniform} and 
Ko and Friedman \cite{ko1988computing}, 
polynomial-time computability of an analytic function 
on a compact interval is 
equivalent to that of its Taylor sequence at a point 
(although the latter is a local property, 
polynomial-time computability on the whole interval is implied 
by analytic continuation; 
see \cite[Corollary~4.5]{muller1987uniform}
or \cite[Theorem~11]{kawamura2010complexity}). 
This implies the polynomial-time computability of $h$, 
since we can efficiently compute the 
Taylor sequence of $h$ from that of $g$. 
} (this does not necessarily mean that 
computing the values of $h$ from those of $g$ is easy; 
see the last paragraph of Section~\ref{section: constructive}). 
Thus our interest is in 
the cases between Lipschitz and analytic 
(the fourth and fifth rows). 
We say that $g$ is of class $\classC ^{(i, j)}$
if the partial derivative $\D ^{(n, m)} g$ 
(often also denoted $\partial ^{n + m} g (t, y) / \partial t ^n \partial y ^m$)
exists and is continuous for all $n \le i$ and $m \le j$;%
\footnote{%
Another common terminology (which we used in the abstract)
is to say that $g$ is of class $\classC ^k$
if it is of class $\classC ^{(i,j)}$ 
for all $i$, $j$ with $i + j \leq k$.}
it is said to be of class $\classC ^{(\infty, j)}$ if
it is of class $\classC ^{(i, j)}$ for all $i \in \N$. 

We will show that adding continuous differentiability does not break the
$\classPSPACE$-completeness that we knew from \cite{kawamura2010lipschitz} 
for the Lipschitz continuous case: 

\begin{theorem}
 \label{DifferentiableIsPspace}
There exists a polynomial-time computable function
$g \colon [0,1] \times [-1,1] \to \R$ 
of class $\classC ^{(\infty, 1)}$ such that
the equation \eqref{eq:ode} has a 
$\classPSPACE$-hard solution $h \colon [0, 1] \to \R$. 
 \end{theorem}

The complexity notions (computability and hardness) in this and the following theorems 
will be explained in Section~\ref{section: preliminaries}. 
When $g$ is more than once differentiable, 
we did not quite succeed in proving that $h$ is $\classPSPACE$-hard
in the same sense, 
but we will prove it $\classCH$-hard, 
where $\classCH \subseteq \classPSPACE$ is the 
counting hierarchy (see Section~\ref{section: counting hierarchy}): 

 \begin{theorem}
  \label{KTimesIsCH}
Let $k$ be a positive integer. 
There is a polynomial-time computable function
$g \colon [0,1] \times [-1,1] \to \R$ 
of class $\classC ^{(\infty, k)}$ such that
the equation \eqref{eq:ode} has a 
$\classCH$-hard solution $h \colon [0, 1] \to \R$. 
 \end{theorem}

In Theorems \ref{DifferentiableIsPspace} and \ref{KTimesIsCH}, 
we said
$g \colon [0,1] \times [-1, 1] \to \R$ instead of 
$g \colon [0,1] \times \R \to \R$, because
the notion of polynomial-time computability of real functions 
in this paper is defined only when the domain is a bounded closed region.%
\footnote{%
Although we could extend our definition to 
functions with unbounded domain~%
\cite[Section~4.1]{kawamura2010operators}, 
the results in Table~\ref{table:related} 
do not hold as they are, 
because polynomial-time compubable functions $g$, 
such as $g (t, y) = y + 1$, 
could yield functions~$h$, such as $h (t) = \exp t - 1$, 
that grow too fast to be polynomial-time (or even polynomial-space) computable. 
Bournez, Gra\c ca and Pouly~%
\cite[Theorem~2]{bournez11:_solvin_analy_differ_equat_in}
report that the statement about the analytic case holds true 
if we restrict the growth of $h$ 
appropriately. 
} 
This makes the equation~\eqref{eq:ode} ill-defined 
in case $h$ ever takes a value outside $[-1, 1]$. 
By saying that $h$ is a solution in Theorem~\ref{DifferentiableIsPspace}, 
we are also claiming that 
$h (t) \in [-1, 1]$ for all $t \in [0, 1]$. 
This is no essential restriction, 
because any pair of functions 
$g \colon [0,1] \times \R \to \R$ and $h \colon [0, 1] \to \R$
satisfying the equation
could be scaled down in an appropriate way 
(without affecting the computational complexity)
to make $h$ stay in $[-1, 1]$. 
In any case, 
since we are making stronger assumptions on $g$ than Lipschitz continuity, 
the solution $h$, if it exists, is unique. 

Whether smoothness of the input function 
reduces the complexity of the output
has been studied for operators other than solving differential equations, 
and the following negative results are known. 
The integral of a polynomial-time computable real function 
can be $\classNumberP$-hard, and this does not change 
by restricting the input to 
$\classC ^\infty$ (infinitely differentiable) functions
\cite[Theorem~5.33]{ko1991complexity}. 
Similarly, the function obtained by maximization 
from a polynomial-time computable real function 
can be $\classNP$-hard, and this is still so
even if the input function is restricted to $\classC ^\infty$ 
\cite[Theorem~3.7]{ko1991complexity}%
\footnote{%
In the last part of the proof of this fact in 
the book \cite[Theorem 3.7]{ko1991complexity}, 
the definition of $f$ needs to be replaced by, e.g., 
\[f(x) = 
\begin{cases}
 u_s & \text{if not } R(s,t), \\
 u_s + 2^{-(p(n)+2n+1)\cdot n} \cdot h_1(2^{p(n)+2n+1} (x - y_{s,t})) & \text{if } R(s,t). 
\end{cases}\]
}. 
(Restricting to analytic inputs 
renders the output polynomial-time computable, 
again by the argument based on Taylor series.)
In contrast, for the differential equation
we only have Theorem~\ref{KTimesIsCH} for each $k$, 
and do not have any hardness result 
when $g$ is assumed to be infinitely differentiable. 

Theorems \ref{DifferentiableIsPspace} and \ref{KTimesIsCH} 
are about the complexity of each solution $h$. 
We will also discuss
the complexity of the final value $h(1)$ and
the complexity of the operator that maps $g$ to $h$; 
see Sections \ref{section: final value} and \ref{section: constructive}. 

\paragraph{Notation}
Let $\N$, $\Z$, $\Q$, $\R$ denote the set of natural numbers,
integers,
rational numbers and 
real numbers, respectively.

We assume that any polynomial is increasing,
since it does not change the meaning of 
polynomial-time computable or polynomial-space computable.

Let $A$ and $B$ be bounded closed intervals in $\R$.
We write $|f| = \sup_{x \in A} f(x)$ for $f \colon A \to \R$.
A function $f \colon A \to \R$ is of \emph{class $\classC^i$}
(or \emph{$i$-times continuously differentiable})
if all the derivatives $\D f, \D^2 f, \dots, \D^i f$ exist and are continuous.

For a function $g$ of two variables, 
we write $\D _1 g$ and $\D _2 g$ for the derivatives of $g$ 
with respect to the first and the second variable,
respectively, 
when they exist.
A function $g \colon A \times B \to \R$ is of \emph{class $\classC^{(i, j)}$}
if for each $m \in \{0, \dots, i\}$ and $n \in \{0, \dots j\}$,
the derivative $\D_1^m \D_2^n g$ exists and is continuous.
This derivative $\D_1^m \D_2^n g$ is then written $\D ^{(m, n)} g$
(and is known to equal $\D _{e_1} \ldots \D _{e _{m + n}} g$ for any 
sequence $e _1 \ldots e _{m + n}$ of $m$ $1$s and $n$ $2$s). 
A function $g$ is of \emph{class $\classC^{(\infty, j)}$}
if it is of class $\classC^{(i, j)}$ for all $i \in \N$.

\section{The Counting Hierarchy}
\label{section: counting hierarchy}

We assume that the reader is familiar with 
the basics of complexity theory, 
such as the classes $\classP$ and $\classPSPACE$
and the notions of reduction and hardness 
\cite{wegener-book}. 
We brief\textcompwordmark ly review the class~$\classCH$. 

The counting hierarchy $\classCH$ is defined analogously to 
the polynomial-time hierarchy $\classPH$ \cite[Kapitel~10.4]{wegener-book} using 
oracle machines corresponding to the class $\classPP$ \cite[Kapitel~3.3]{wegener-book}
instead of $\classNP$: 
thus, $\classCH = \bigcup _{n \in \N} \quantC _n \classP$, 
where each level $\quantC _n \classP$ is defined inductively by $
 \quantC_0 \classP  = \classP 
$ and $
 \quantC_{n+1} \classP = \classPP^{\quantC_n \classP}
$. 
We leave the details of this definition to 
the original papers~%
\cite{wagner1986complexity,parberry-schnitger,toran1991complexity}. 
All we need for our purpose is the fact (Lemma~\ref{lemma:CnP-complete} below) 
that $\quantC _n \classP$ 
is characterized by the following complete problem $\quantC_n B_{\mathrm{be}}$:
given $n$ lists $X _1$, \ldots, $X _n$ of propositional variables, 
a propositional formula $\phi (X _1, \dots, X _n)$ with all free variables listed, 
and numbers $m _1$, \ldots, $m _n$ in binary, 
determine whether 
\begin{equation}
 \quantC^{m_n}{X_n} \ldots \quantC^{m_1}{X_1} \phi(X_1, \dots, X_n)
\end{equation}
is true.  Here, $\quantC ^m$ is the \emph{counting quantifier}: 
for every formula $\phi(X)$ with the list $X$ of $l$ free variables,
we write 
\begin{equation}
 \quantC^m X \phi(X) 
  \iff 
  \sum_{x \in \{0,1\}^l} \phi(x) \ge m,
\end{equation}
where formula $\phi$ is identified with the function 
$\phi \colon \{0,1\}^l \to \{0,1\}$
such that $\phi(X) = 1$ when $\phi(X)$ is true.

\begin{lemma}[{\cite[Theorem 7]{wagner1986complexity}}] \label{lemma:CnP-complete}
 For every $n \ge 1$, 
 the problem $\quantC_n B_{\mathrm{be}}$ is $\quantC_n\classP$-complete.
\end{lemma}

By this, we mean that 
any problem $A$ in $\quantC_n \classP$ reduces to $\quantC_n B_{\mathrm{be}}$ 
via some polynomial-time function $F$ 
in the sense that $v \in A$ if and only if $F (v) \in \quantC_n B_{\mathrm{be}}$. 

Note that this is analogous to the polynomial hierarchy $\classPH$, 
whose $n$th level $\classSigma ^p _n$ has the complete problem
that asks for the value of a formula of the form $
\exists X _n \forall X _{n - 1} \exists X _{n - 2} \ldots \allowbreak \phi (X_1, \dots, X_n)
$ with $n$ alternations of universal and existential quentifiers. 
Note also that the counting quantifier generalizes 
the universal and existential quentifiers, and hence $\classPH \subseteq \classCH$ -- 
in fact, it is known~\cite{toda} that 
$\classPH$ is contained already in $\classP ^{\classPP}$
and thus in the second level of the counting hierarchy. 

\section{Computational Complexity of Real Functions}
\label{section: preliminaries}

We review some complexity notions 
from Computable Analysis~%
\cite{ko1991complexity,weihrauch00:_comput_analy}. 
We start by fixing an encoding of real numbers 
by string functions.

\begin{definition}
 A function $\phi \colon \{0, 1\} ^* \to \{0, 1\}^*$ is a \emph{name} of a real number $x$ 
 if for all $n \in \N$,
 $\phi(0^n)$ is the binary representation of $\lfloor x \cdot 2^n \rfloor$ or
 $\lceil x \cdot 2^n \rceil$,
 where $\lfloor \cdot \rfloor$ and $\lceil \cdot \rceil$ mean
 rounding down and up to the nearest integer.
 \end{definition}

In effect, a name of a real number $x$ takes each string $0 ^n$ 
to an approximation of $x$ with precision $2 ^{-n}$.

We use \emph{oracle Turing machines} (henceforth just \emph{machines})
to work on these names (Fig.~\ref{fig:model-of-function}).
Let $M$ be a machine and $\phi$ be a function from strings to strings. 
We write $M ^\phi (0 ^n)$ for the output string 
when $M$ is given
$\phi$ as oracle and string $0^n$ as input.
Thus we also regard $M^\phi$ as a function from strings to strings.

\begin{figure}
 \begin{center}
  \includegraphics[height=0.17\textheight]{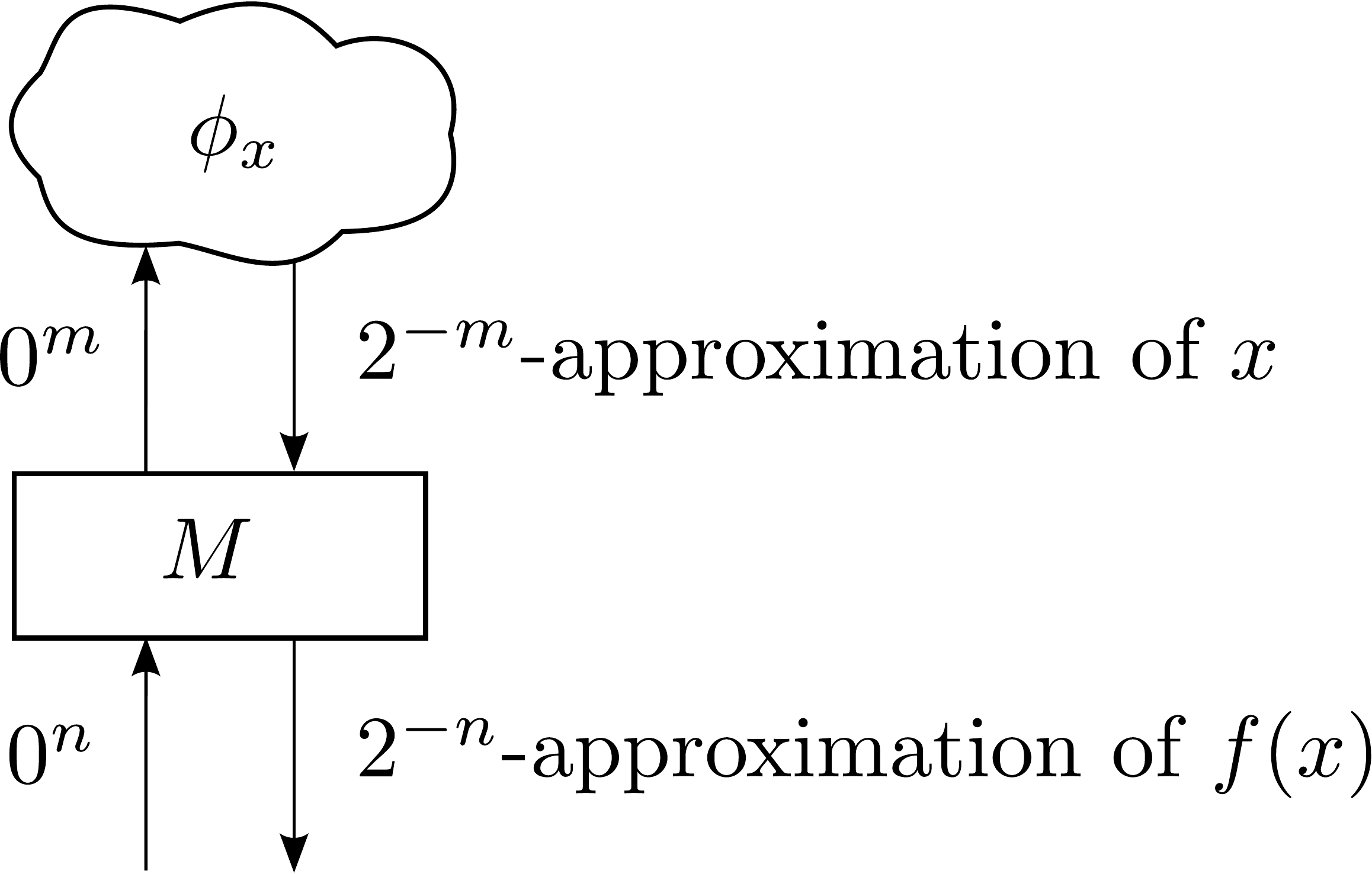}
 \end{center}
 \caption{A machine $M$ computing a real function $f$}
 \label{fig:model-of-function}
\end{figure}

\begin{definition}
Let $A$ be a bounded closed interval of\/ $\R$.
A machine $M$ \emph{computes} a real function $f \colon A \to \R$ 
if for any $x \in A$ and any name $\phi_x$ of it,
$M^{\phi_x}$ is a name of $f(x)$.
\end{definition}

Computation of a function $f \colon A \to \R$ on
a two-dimensional bounded closed region $A \subseteq \R ^2$ 
is defined similarly using machines with two oracles.

A real function is (\emph{polynomial-time}) \emph{computable} if there exists some machine that computes it (in polynomial time).
Polynomial-time computability of a real function $f$ means that
for any $n \in \N$, 
an approximation of $f(x)$ with error bound $2^{-n}$
is computable in time polynomial in $n$ 
independent of the real number $x$.

By the time the machine outputs the approximation of $f (x)$ of precision~$2 ^{-n}$, 
it knows $x$ only with some precision $2 ^{-m}$.
This implies that all computable real functions are continuous.
If the machine runs in polynomial time,
this $m$ is bounded polynomially in $n$.
This implies \eqref{eq:modulus} in the following lemma, 
which characterizes polynomial-time real functions
by the usual polynomial-time computability of string functions 
without using oracle machines. 

\begin{lemma}
 \label{lem:type1representation}
 A real function is polynomial-time computable if and only if
 there exist a polynomial-time computable function 
 $\phi \colon (\Q \cap [0, 1]) \times \{0, 1\} ^* \to \Q$ and 
 polynomial $p \colon \N \to \N$ such that
 for all $d \in \Q \cap [0,1]$ and $n \in \N$,
 \begin{equation}
  |\phi(d, 0^n) - f(d)| \le 2^{-n},
 \end{equation}
 and for all $x, y \in [0, 1]$, $n \in \N$,
 \begin{equation} 
  |x-y| \le 2^{-p(n)} \Rightarrow |f(x) - f(y)| \le 2^{-n},
   \label{eq:modulus}
 \end{equation}
where each rational number is written
as a fraction whose numerator and denominator
are integers in binary.
\end{lemma}

To speak about hardness of a real function, 
we define the notion of a language to it. 
A language $L \subseteq \{0, 1\} ^*$ is identified with the function 
$L \colon \{0, 1\} ^* \to \{0, 1\}$ such that $L (u) = 1$ when $u \in L$.

\begin{definition} \label{definition: reduction}
 A language $L$ \emph{reduces} to a function $f \colon [0, 1] \to \R$
 if there exist a polynomial-time function $S$ and 
 a polynomial-time oracle Turing machine $M$ (Fig.~\ref{fig:reduction})
 such that for any string $u$, 
  \begin{enumerate}
   \item $S (u, \cdot)$ is a name of a real number $x _u$, and 
   \item for any name $\phi$ of $f (x _u)$, we have that 
    $M ^\phi (u)$ accepts if and only if $u \in L$.
  \end{enumerate}
\end{definition}
This reduction may superficially look stronger (and hence the reducibility weaker) than
the one in Kawamura~\cite{kawamura2010lipschitz} 
in that $M$ can make multiple queries adaptively, 
but it is not hard to see that this makes no difference. 

For a complexity class~$C$, a function $f$ is \emph{$C$-hard}
if all languages in $C$ reduce to $f$.

 \begin{figure}
  \begin{center}
  \includegraphics[scale=0.25]{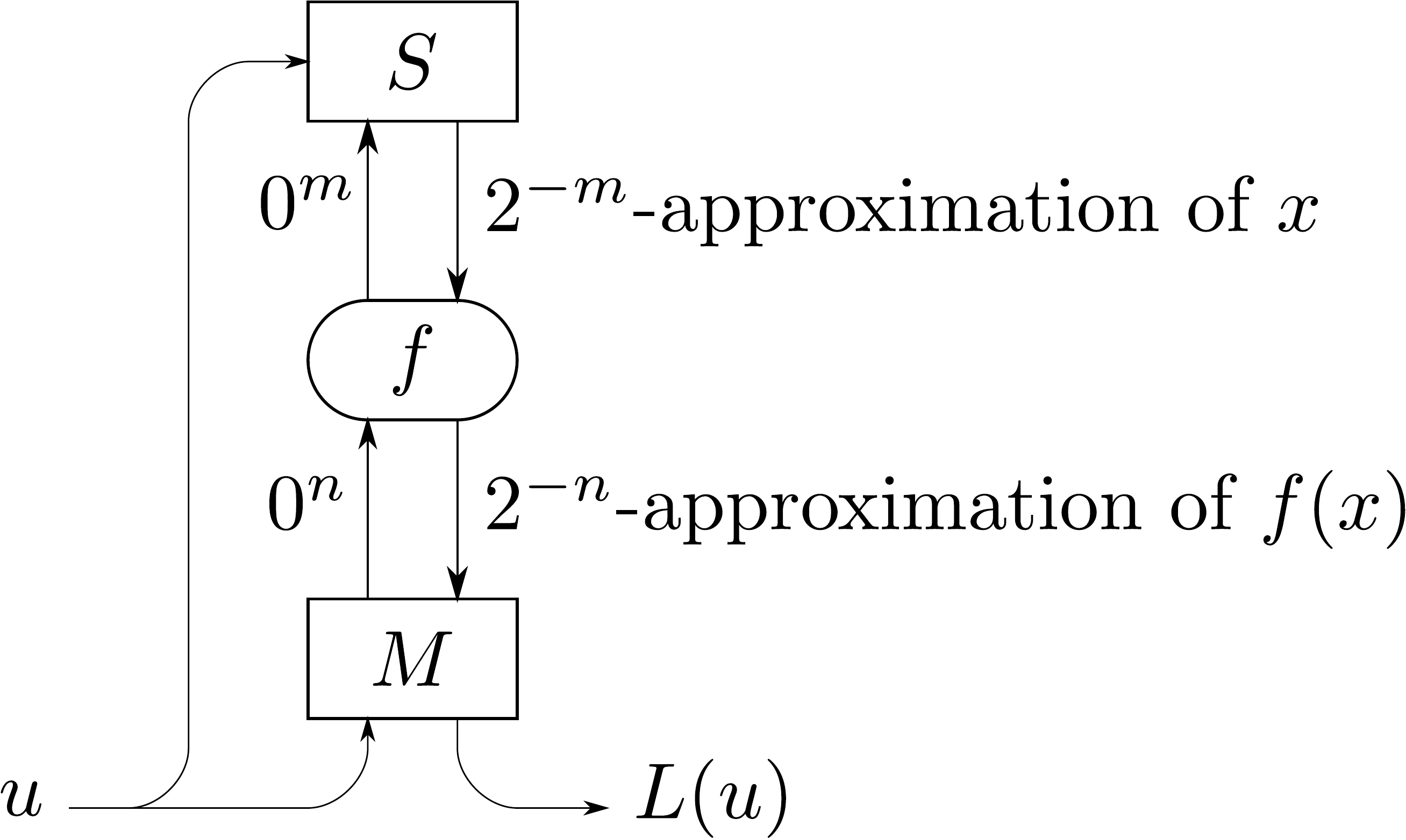}
  \caption{Reduction from a language $L$ to a function $f \colon [0,1] \to \R$}
  \label{fig:reduction}
  \end{center}
 \end{figure}

\section{Proof of the Main Theorems}
\label{section:differentiable}

The proofs of Theorems \ref{DifferentiableIsPspace} and \ref{KTimesIsCH}
proceed as follows. 
In Section~\ref{section:divp}, 
we define \emph{difference equations}, 
a discrete version of the differential equations. 
We show that these equations with certain restrictions 
are $\classPSPACE$- and $\classCH$-hard. 
In Section~\ref{subsection: ode family}, 
we show that these classes of difference equations can be simulated
by families of differential equations 
satisfying certain uniform bounds on higher-order derivatives. 
In Section~\ref{subsection: proof of theorems}, 
we prove Theorems \ref{DifferentiableIsPspace} and \ref{KTimesIsCH} by
putting these families of functions together 
to obtain one differential equation 
with the desired smoothness 
($\classC ^{(\infty, 1)}$ and $\classC ^{(\infty, k)}$). 

The idea of simulating a discrete system with limited feedback 
by differential equations
was essentially already present in the proof of 
the Lipschitz version of these results \cite{kawamura2010lipschitz}. 
We look more closely at this limitation on feedback, 
and express it as a restriction on the \emph{height}
of the difference equation. 
We show that a stronger height restriction allows
the difference equation to be simulated by 
smoother differential equations 
(see the proof of Lemma~\ref{KTimesFamily} and the discussion after it), 
leading to the $\classCH$-hardness for $\classC ^{(\infty, k)}$ functions.

\subsection{Difference equations}
\label{section:divp}

In this section, we define difference equations, 
which are a discrete version of differential equations. 
Then we show the $\classPSPACE$- and $\classCH$-hardness of 
families of difference equations with different height restrictions. 

Let $[n]$ denote $\{0, \dots , n-1\}$.
Let $G \colon [P] \times [Q] \times [R] \to \{-1, 0, 1\}$ and
$H \colon [P + 1] \times [Q+1] \to [R]$. 
We say that $H$ is the solution of the \emph{difference equation} given by $G$
if for all $i \in [P]$ and $T \in [Q]$ (Fig.~\ref{fig:divp}), 
\begin{gather}
   H(i, 0) = H(0, T) = 0, \label{eq:initial value}
\\
   H(i + 1, T + 1) - H(i+1, T) = G(i, T, H(i, T)).  \label{eq:divp}
\end{gather}
We call $P$, $Q$ and $R$ the \emph{height}, \emph{width} and \emph{cell size} of
the difference equation.
The equations \eqref{eq:initial value} and \eqref{eq:divp} are similar to 
the initial condition $h(0) = 0$ and the equation $\D h(t) = g(t, h(t))$ 
in \eqref{eq:ode}, respectively.
In Section~\ref{subsection: ode family}, 
we will use this similarity to simulate difference equations by differential equations.

\begin{figure}
 \begin{center}
  \includegraphics[height=0.15\textheight]{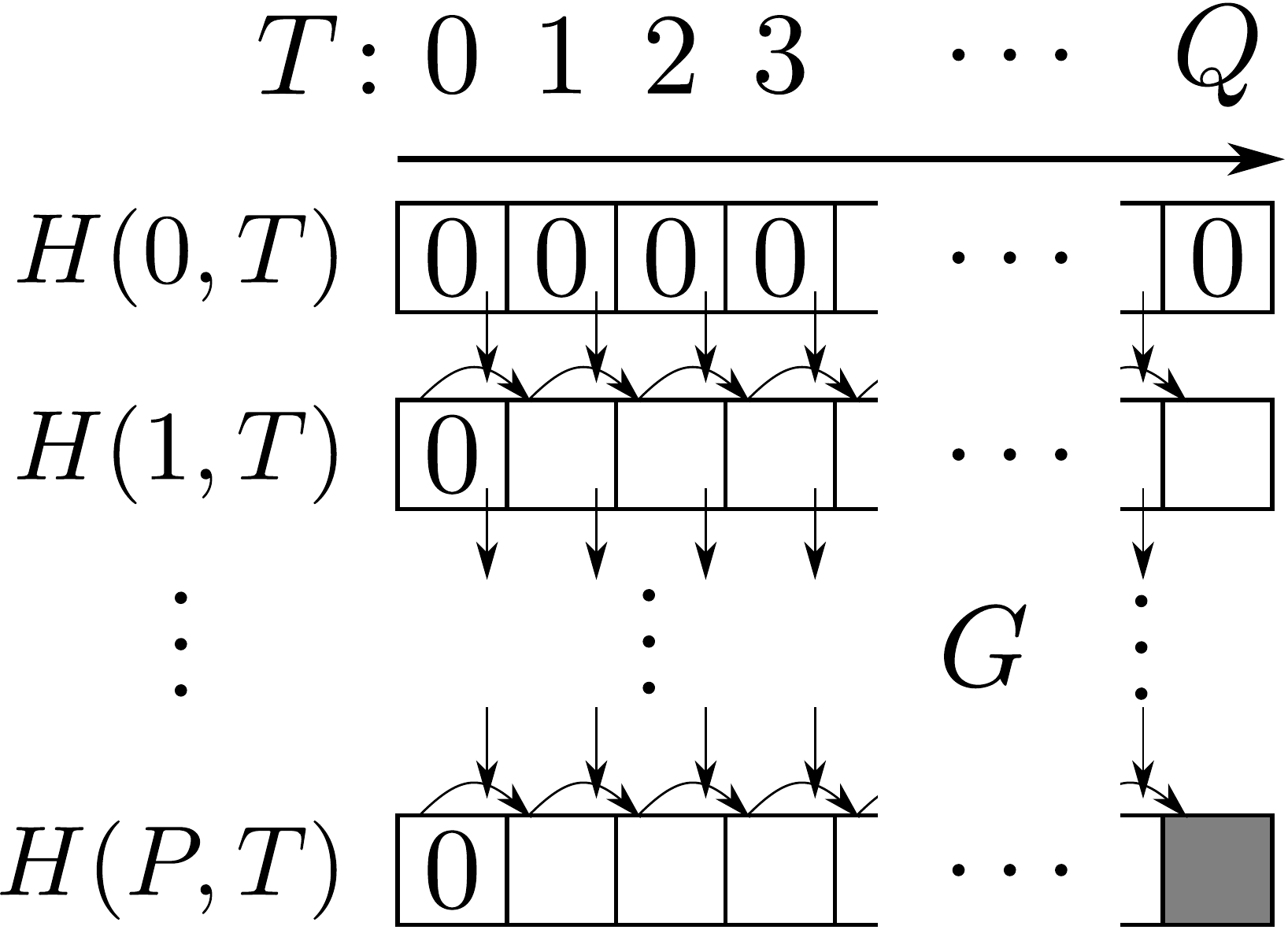}
 \end{center}
 \caption{The solution $H$ of the difference equation given by $G$}
 \label{fig:divp}
\end{figure}

We view a family of difference equations as a computing system by
regarding the value of the bottom right cell (the gray cell in Fig.~\ref{fig:divp}) as the output. 
A family $(G_u)_u$ of functions 
$G_u \colon [P_u] \times [Q_u] \times [R_u] \to \{-1, 0, 1\}$
\emph{recognizes} a language $L$ if for each $u$,
the difference equation given by $G_u$ has a solution $H_u$ 
and $H_u(P_u, Q_u) = L(u)$.
A family $(G_u)_u$ is \emph{uniform} 
if the height, width and cell size of $G_u$ are 
polynomial-time computable from $u$ (in particular, 
they must be bounded by $2^{p(|u|)}$, for some polynomial~$p$)
and $G_u(i, T, Y)$ is polynomial-time computable from $(u, i, T, Y)$.
A family $(G_u)_u$ has \emph{polynomial height} 
(and \emph{logarithmic height}, respectively)
if the height $P_u$ is bounded by $|u| ^{\mathrm O (1)}$
(and by $\mathrm O (\log |u|)$, respectively).
With this terminology,
the key lemma in 
\cite[Lemma 4.7]{kawamura2010lipschitz} 
can be written as follows:
\begin{lemma}
 \label{DIVPpolyIsPSPACEhard}
 There exists a $\classPSPACE$-hard language $L$ that is recognized by some uniform family of functions with polynomial height%
 \footnote{In fact, 
this language~$L$ is in $\classPSPACE$, 
because a uniform family with polynomial height 
can be simulated in polynomial space.
}.
\end{lemma}

Kawamura~\cite{kawamura2010lipschitz} obtained the hardness result in the third row of Table \ref{table:related} 
by simulating the difference equations of Lemma~\ref{DIVPpolyIsPSPACEhard}
by Lipschitz continuous differential equations. 
Likewise, 
Theorem~\ref{DifferentiableIsPspace} follows from Lemma~\ref{DIVPpolyIsPSPACEhard}
by a modified construction that keeps 
the function in class $\classC ^{(\infty, 1)}$ 
(Sections \ref{subsection: ode family} and \ref{subsection: proof of theorems}).

We show further that 
difference equations restricted to logarithmic height can be simulated by
$\classC ^{(\infty, k)}$ functions for each $k$ 
(Sections \ref{subsection: ode family} and \ref{subsection: proof of theorems}).
Theorem~\ref{KTimesIsCH} follows from this simulation and the following lemma.
\begin{lemma}
 \label{DIVPlogIsCHhard}
 There exists a $\classCH$-hard language $L$ that is recognized by some uniform family of functions with logarithmic height.
\end{lemma}

\begin{proof}
We define the desired language~$L$ by
\begin{equation}
\label{equation: definition of padded CQBF}
 \langle 0^{2^n}, w \rangle \in L
 \iff
 w \in \quantC_n B_{\mathrm{be}}, 
\end{equation}
using the $\quantC _n \classP$-hard language $\quantC_n B_{\mathrm{be}}$ 
from Lemma~\ref{lemma:CnP-complete}. 
Then $L$ is obviously $\classCH$-hard. 

We construct a logarithmic-height uniform function family $(G_u)_u$
recognizing $L$.
We will describe how to construct $G _u$ 
for a string~$u$ of the form $\langle 0^{2^n}, \allowbreak
\langle \phi(X_1, \dots, X_n), m_1, \dots, m_n \rangle \rangle$, 
where 
$\phi$ is a formula, and 
$X _1$, \ldots, $X _n$ are lists of variables
with lengths $m _1$, \ldots, $m _n$, respectively. 
 
We write $l_i = |X_i|$ and $s_i = \sum^i_{j=1} (l _j + 1)$.
For each $i \in \{0, \dots, n\}$ and
$x _{i+1} \in \{0,1\}^{l_{i+1}}, \dots, x _n \in \{0,1\}^{l_n}$,
we write $\phi_i(x_{i+1}, \dots, x_n)$ for the truth value ($1$ or $0$) of $
\quantC^{m_{i}}{X_i} \ldots \allowbreak \quantC^{m_1}{X_1} \allowbreak 
\phi(X_1, \dots, X_i, x_{i+1}, \dots, x_n)
$.  Note that 
$\phi _0 (x _1, \ldots, x _n) = \phi (x _1, \ldots, x _n)$, 
$\phi _n () = L (u)$, 
and the relation between $\phi _{i - 1}$ and $\phi _i$ is given by 
\begin{equation}
\label{eq:phi-step}
   \phi _i (x _{i+1}, \dots, x_n) 
  = 
   C^{m_i} \Bigl( 
     \sum_{x_i \in \{ 0,1 \} ^{l_i}}
		\phi _{i - 1} (x _i, x_{i+1}, \dots, x_{n})
   \Bigr), 
\end{equation}
where we define $C ^m \colon \N \to \{0,1\}$ by 
\begin{equation}
 C ^m (k) 
  = \begin{cases}
     1 & \text{if} \ k \ge m, \\
     0 & \text{if} \ k < m.
    \end{cases}
\end{equation}
For $T \in \N$, we write $T_i$ for the 
$i$th digit of $T$ written in binary 
($T _0$ being the least significant digit),
and $T_{[i,j]}$ for the string $T_{j-1} T_{j-2} \cdots T_{i+1} T_{i}$.

For each $(i, T, Y) \in [n+1] \times [2^{s_n}+1] \times [2^{|u|}]$,
we define $G_u (i, T, Y)$ as follows.
The first row is given by
\begin{equation}
\label{eq:def-Gu:case0}
  G_u(0,T,Y) = 
   (-1)^{T_{s_1}}\phi(T_{[1,s_1]}, T_{[s_1+1,s_2]},
    \dots, T_{[s_{n-1}+1,s_n]}), 
\end{equation}
and for $i \neq 0$, we define 
\begin{equation} 
  G_u(i,T,Y) = 
   \begin{cases}
    (-1)^{T_{s_{i+1}}} C^{m_i}(Y) 
    & \text{if} \ T_{[1,s_i+1]} = 10 \cdots 0, \\
    0 & \text{otherwise}.
   \end{cases} 
   \label{eq:def-Gu:case-nonzero}
 \end{equation}
This family $(G_u)_u$ is clearly uniform, 
and its height $n + 1$ is logarithmic in $|u|$.

Let $H_u$ be the solution (as defined in 
\eqref{eq:initial value} and \eqref{eq:divp}) 
of the difference equation given by $G_u$. 
We prove by induction on $i$ 
that $H _u (i, T) \in [2^{l_i} + 1]$ for all $T$, and that 
\begin{equation} \label{eq:subformula}
  G_u(i,V,H_u(i,V)) = (-1)^{V_{s_{i+1}}} 
   \phi_i(V_{[s_i+1, s_{i+1}]}, \dots, V_{[s_{n-1}+1, s_n]})
\end{equation}
if $V_{[1, s_i +1]} = 10 \cdots 0$
(otherwise it is immediate from the definition that $G_u(i, V, H_u(n, V)) = 0$).
For $i = 0$, the claims follow from \eqref{eq:def-Gu:case0}.
For $i > 0$, suppose that the induction hypothesis
\begin{equation} \label{eq:subformula minus one}
  G _u (i - 1, V, H _u (i - 1, V)) = (-1) ^{V_{s_i}} 
   \phi _{i - 1} (V _{[s _{i - 1} + 1, s _i]}, \dots, V _{[s_{n-1}+1, s_n]}) 
\end{equation}
holds. 
Since $H _u$ is the solution of the difference equation given by $G _u$, we have
 \begin{equation} \label{eq:summation}
  H _u (i, T) 
  = \sum_{V = 0}^{T-1} G_u(i - 1, V, H_u(i - 1, V)).
 \end{equation}
Since the assumption \eqref{eq:subformula minus one} 
implies that flipping the bit $V_{s_i}$ of
any $V$ reverses the sign of $G _u (i - 1, V, H _u(i - 1, V))$,
most of the summands in \eqref{eq:summation} cancel out.
The only nonzero terms that can survive 
are the ones corresponding to those $V$ 
that satisfy $V_{[1, s_{i - 1}+1]} = 10 \cdots 0$ and
lie between 
the numbers whose binary representations are
$T_{s_n} \dots T_{s_i+1} 00 \dots 0$ and 
$T_{s_n} \dots T_{s_i+1} 01 \dots 1$. 
There are at most $2 ^{l _i}$ such terms, 
and each of them is $0$ or $1$, 
so $H _u (i, T) \in [2 ^{l _i} + 1]$.
In particular, if $V_{[1,s_i+1]} = 10 \cdots 0$, we have
 \begin{equation}
  H _u (i, V) = \sum_{x \in \{0,1\}^{l_i}}
  \phi _{i - 1} (x, V_{[s_i+1, s_{i+1}]}, \dots, V_{[s_{n-1}+1, s_n]}).
 \end{equation}
By this and \eqref{eq:phi-step}, we have 
 \begin{equation}
  C^{m_i} (H _u(i, V)) 
  = \phi _i (V_{[s_i+1, s_{i+1}]}, \dots, V_{[s_{n-1}+1, s_n]}).
\end{equation}
By this and \eqref{eq:def-Gu:case-nonzero}, we have \eqref{eq:subformula}, 
completing the induction step.

By substituting $n$ for $i$ and $2^{s_n}$ for $V$ in \eqref{eq:subformula},
we get $G_u(n, 2^{s_n}, H_u(n,2^{s_n})) = \phi_n() = L (u)$.
Hence $H_u(n+1, 2^{s_n}+1) = L (u)$.
\end{proof}

Note that for $\classCH$-hardness, 
we could have defined $L$ using a faster growing function than $2 ^n$ in
\eqref{equation: definition of padded CQBF}. 
This would allow the difference equation in Lemma~\ref{DIVPlogIsCHhard}
to have height smaller than logarithmic. 
We stated Lemma~\ref{DIVPlogIsCHhard} with logarithmic height, 
simply because this was the highest possible difference equations
that we were able to simulate by $\classC ^{(\infty, k)}$ functions
(in the proof of Lemma~\ref{KTimesFamily} below). 

\subsection{Simulating difference equations by real functions}
\label{subsection: ode family}
We show that certain families of smooth differential equations can simulate 
$\classPSPACE$- and $\classCH$-hard difference equations from the previous section.

Before stating Lemmas~\ref{KTimesFamily} and \ref{DifferentiableFamily},
we extend the definition of polynomial-time computability of real functions
to that of \emph{families} of real functions.
A machine $M$ \emph{computes} a family $(f_u)_u$ of functions $f _u \colon A \to \R$ 
indexed by strings $u$
if for any $x \in A$ and any name $\phi_x$ of $x$,
the function taking $v$ to $M ^{\phi _x} (u, v)$ is a name of $f _u (x)$.
We say a family of real functions $(f_u)_u$ is polynomial-time if there is
a polynomial-time machine computing $(f_u)_u$.

 \begin{lemma}
  \label{KTimesFamily}
  There exist a $\classCH$-hard language $L$ and a polynomial $\mu$,
  such that for any $k \ge 1$ and polynomial $\gamma$,
  there are a polynomial $\rho$ and families $(g_u)_u$, $(h_u)_u$ of real functions
  such that $(g_u)_u$ is polynomial-time computable and for any string $u$:
  \begin{enumerate}
   \item \label{enum:kf:start}
	 $g_u\colon [0,1] \times [-1,1]\to \R$, $h_u\colon [0,1] \to [-1,1]$;
   \item \label{enum:equation}
	 $h_u(0) = 0$ and $\D h_u(t) = g_u(t, h_u(t))$ for all $t \in [0,1]$;
   \item \label{enum:differentiability}
         $g_u$ is of class $\classC^{(\infty, k)}$;
   \item \label{enum:boundary}
	 $
	 \D^{(i, 0)} g_u(0,y) = \D^{(i, 0)} g_u(1,y) = 0
         $ for all $i \in \N$ and $y \in [-1,1]$;
   \item \label{enum:smooth}
	 $
	 \left|\D^{(i,j)} g_u(t,y)\right| \le 2^{\mu(i, |u|) - \gamma(|u|)}
         $ for all $i \in \N$ and $j \in \{0, \dots, k\}$;
   \item \label{enum:kf:end}
	 $h_u(1) = 2^{-\rho(|u|)} L(u)$.
  \end{enumerate}
 \end{lemma}

\begin{lemma}
 \label{DifferentiableFamily}
 There exist a $\classPSPACE$-hard language $L$ and a polynomial $\mu$,
 such that for any polynomial $\gamma$,
 there are a polynomial $\rho$ and families $(g_u)_u$, $(h_u)_u$ of real functions
 such that $(g_u)_u$ is polynomial-time computable and for any string $u$
 satisfying (\ref{enum:kf:start})--(\ref{enum:kf:end}) of Lemma~\ref{KTimesFamily} with $k = 1$.
\end{lemma}

In Lemmas~\ref{KTimesFamily} and \ref{DifferentiableFamily}, 
we have the new conditions (\ref{enum:differentiability})--(\ref{enum:smooth})
about the smoothness and the derivatives of $g_u$ 
that were not present in \cite[Lemma 4.1]{kawamura2010lipschitz}.
To satisfy these conditions,
we construct $g_u$ using the smooth function $f$ in following lemma.

\begin{lemma}[{\cite[Lemma 3.6]{ko1991complexity}}]
 \label{SmoothFunction}
 There exist 
 a polynomial-time function $f \colon [0, 1] \to \R$ of class $\classC^\infty$ and 
 a polynomial $s$ such that
  \begin{itemize}
   \item $f(0) = 0$ and $f(1) = 1$;
   \item $\D ^n f (0) = \D ^n f (1) = 0$ for all $n \ge 1$;
   \item $f$ is strictly increasing;
   \item $\D ^n f$ is polynomial-time computable for all $n \ge 1$;
   \item $|\D^n f| \le 2^{s(n)}$ for all $n \ge 1$. 
  \end{itemize}
 \end{lemma}

Although the existence of the polynomial~$s$ 
was not explicitly stated in \cite[Lemma 3.6]{ko1991complexity},
it can be proved easily.

We will prove Lemma~\ref{KTimesFamily} using Lemma~\ref{DIVPlogIsCHhard} as follows.
Let $(G_u)_u$ be a family of functions obtained by Lemma~\ref{DIVPlogIsCHhard},
and let $(H_u)_u$ be the family of the solutions of 
the difference equations given by $(G_u)_u$.
We construct $h_u$ and $g_u$ from $H_u$ and $G_u$ 
so that $h_u(T/2^{q(|u|)}) = \sum^{p(|u|)}_{i = 0} H_u(i, T)/B^{d_u(i)}$ for each $T = 0$, \ldots, $2^{q(|u|)}$
and $\D h_u(t) = g_u(t, h_u(t))$.
The polynomial-time computability of $(g_u)_u$ follows from that of $(G_u)_u$.
We omit the analogous and easier proof of Lemma~\ref{DifferentiableFamily}.

\begin{proof}[Proof of Lemma~\ref{KTimesFamily}]
Let $L$ and $(G_u)_u$ be as in Lemma~\ref{DIVPlogIsCHhard},
and let $H_u$ be the solutions of 
the difference equations given by $G_u$.
Let $f$ and $s$ be as in Lemma~\ref{SmoothFunction}.

We may assume that 
$G_u \colon [p(|u|)] \times [2^{q(|u|)}] \times [2^{r(|u|)}] \to \{-1, 0, 1\}$
for some $p$, $q$, $r \colon \N \to \N$, 
where $p$ has logarithmic growth and $q$ and $r$ are polynomials. 
We may also assume, similarly to the beginning of the proof of \cite[Lemma 4.1]{kawamura2010lipschitz},
that 
\begin{gather}
 H_u(i, 2^{q(|u|)}) = \begin{cases}
		       L(u) & \text{if} \ i=p(|u|), \\
		       0 & \text{if} \ i<p(|u|), 
		      \end{cases}
\\
 G_u(i, T, Y) \neq 0 \to i = j_u(T)
\end{gather}
for some $
j _u \colon [2 ^{q (\lvert u \rvert)}] \to [p (\lvert u \rvert)]
$. 

We construct families of real functions $(g_u)_u$ and $(h_u)_u$ 
that simulate $G _u$ and $H _u$ 
in the sense that $h_u(T/2^{q(|u|)}) = \sum^{p(|u|)}_{i = 0}H_u(i, T)/B^{d_u(i)}$, 
where the constant $B$ and the 
(increasing) function $d_u \colon [p(|u|)+1] \to \N$ are 
defined by
  \begin{align}
   \label{eq:positioning}
   B &= 2^{\gamma(|u|) + r(|u|) + s(k) + k + 3}, 
   &
   d_u(i) &= 
   \begin{cases}
    (k+1)^i & \text{if} \ i<p(|u|), 
    \\
    \sigma(|u|) & \text{if} \ i=p(|u|), 
   \end{cases}
  \end{align}
where $\sigma$ is a polynomial satisfying $(k+1)^{p(x)} \le \sigma(x)$
(which exists because $p$ is logarithmic). 
Thus, the value $H _u (i, T) \in [2 ^{r (\lvert u \rvert)}]$ will be stored 
in the $d _u (i)$th digit of the base-$B$ expansion of 
the real number $h _u (T / 2 ^{q (\lvert u \rvert)})$. 
This exponential spacing described by $d _u$ will be needed
when we bound the derivative $\D ^{(i, j)} g _u$ 
in \eqref{equation: upperbound of differences} below. 

For each $(t, y) \in [0,1] \times [-1, 1]$,
there exist unique $N \in \N$, $\theta \in [0,1)$, $Y \in \Z$ and $\eta \in [-1/4, 3/4)$
such that $t = (T + \theta)2^{-q(|u|)}$ and $y = (Y + \eta)B^{-d_u(j_u(T))}$.
Using $f$ and a polynomial $s$ of Lemma~\ref{SmoothFunction},
we define 
$\tilde g_{u,Y} \colon [0,1] \to \R$,
$
g _u \colon [0, 1] \times [-1, 1] \to \R
$ and $
h _u \colon [0, 1] \to [-1, 1]
$ by
  \begin{align}
    \label{eq:delta}
   \tilde g_{u, Y} (t) &= \frac{2^{q(|u|)} \D f(\theta)}{B^{d_u(j_u(T)+1)}} 
   G_u \bigl( j_u(T), T, Y \bmod 2^{r(|u|)} \bigr),
   \\
  \label{eq:gu}
  g_u(t,y) 
  &= \begin{cases}
     \tilde g_{u, Y}(t)
     & \text{if } \eta \le \frac{1}{4}, 
     \\
     ( 1-f ( \frac{4\eta-1}{2})) \tilde g_{u, Y}(t)
     + f ( \frac{4\eta-1}{2}) \tilde g_{u,Y+1}(t)
     & \text{if } \eta > \frac{1}{4},
    \end{cases}
   \\
  h_u(t) 
   &= \sum^{p(|u|)}_{i=0} \frac{H_u(i, T)}{B^{d_u(i)}}  
  + \frac{f(\theta)}{B^{d_u(j_u(T)+1)}} G_u \bigl( j_u(T), T, H_u(j_u(T), T) \bigr).
  \label{eq:hu}
  \end{align}

We will verify that $(g_u)_u$ and $(h_u)_u$ defined above satisfy all the conditions stated in Lemma~\ref{KTimesFamily}.
Polynomial-time computability of $(g_u)_u$ can be verified using Lemma~\ref{lem:type1representation}.
The condition~(\ref{enum:kf:start}) is immediate, 
and (\ref{enum:equation}) follows from the relation between $G _u$ and $H _u$ 
(by a similar argument to \cite[Lemma 4.1]{kawamura2010lipschitz}).

Since $g _u$ is constructed by interpolating between 
the (finitely many) values of $G _u$ 
using a $\classC^{(\infty, \infty)}$ function~$f$, 
it is of class $\classC^{(\infty, \infty)}$ 
and thus satisfies (\ref{enum:differentiability}).
Calculating from \eqref{eq:delta}, 
we have 
\begin{equation}
 \label{eq: derivatives of delta}
    \D^i \tilde g_{u,Y}(t) 
    = \frac{2^{(i+1)q(|u|)} \D^{i+1}f(\theta)}{B^{d_u(j_u(T)+1)}}
    G_u\bigl( j_u(T), T, Y \bmod 2^{r(|u|)} \bigr) 
\end{equation}
for each $i \in \N$. 
By this and \eqref{eq:gu}, 
we have
\begin{equation}
     \D^{(i, 0)} g_u(t, y)
     = \begin{cases}
 	\D^i \tilde g_{u, Y}(t) 
	& \text{if} \ \eta \le \frac 1 4, 
	\\
	\bigl( 1-f \bigl( \frac{4\eta-1}{2} \bigr) \bigr) 
	\D^i \tilde g_{u, Y}(t)
	+ f \bigl( \frac{4\eta-1}{2} \bigr) \D^i \tilde g_{u,Y+1}(t) 
	& \text{if} \ \frac 1 4 < \eta 
       \end{cases} \label{eq:d(i,0)g_u}
\end{equation}
for each $i \in \N$ and 
  \begin{equation} \label{eq:d(i,j)g_u}
    \D^{(i, j)} g_u(t, y)
     = \begin{cases}
	0 & \text{if} \ {-\frac 1 4} < \eta < \frac 1 4, \\
	\bigl( 2B^{d_u(j_u(T))} \bigr)^j \cdot \D^j f \bigl( \frac{4\eta - 1}2 \bigr) \cdot
	(\D^i \tilde g_{u,Y+1}(t)-\D^i \tilde g_{u, Y}(t)) 
	& \text{if} \ \frac 1 4 < \eta < \frac 3 4 
       \end{cases}
  \end{equation}
for each $i \in \N$ and $j \in \{1, \dots, k\}$. 
Substituting $t = 0, 1$ ($\theta = 0$) into \eqref{eq:d(i,0)g_u},
we get $\D^{(i, 0)} g_u(0,y) = \D^{(i, 0)} g_u(1,y) = 0$, 
as required in (\ref{enum:boundary}).

We show that (\ref{enum:smooth}) holds with $\mu(x, y) = (x+1)q(y) + s(x+1)$.
Note that $\mu$ is a polynomial independent of $k$ and $\gamma$, 
and that $
 |\D^i \tilde g_{u,Y}(t)| 
\leq
 2 ^{\mu (i, \lvert u \rvert)} / B ^{d _u (j _u (T) + 1)}
$ by \eqref{eq: derivatives of delta}. 
Hence, 
\begin{align}
\label{equation: upperbound of differences}
  |\D^{(i,j)} g _u (t, y)| 
&
 \le 
    \bigl( 2 B^{d_u (j _u (T))} \bigr) ^k \cdot 2^{s(k)}
   \cdot 
    2 \cdot \frac{2 ^{\mu (i, \lvert u \rvert)}}{B ^{d _u (j _u (T) + 1)}} 
 =
   \frac{2 ^{\mu (i, \lvert u \rvert)} \cdot 2^{s(k) + k + 1}}{B ^{d _u (j _u (T) + 1) - k \cdot d _u (j _u (T))}} 
\notag
\\
&
 \le
   2 ^{\mu (i, \lvert u \rvert)} \cdot \frac{2^{s(k) + k + 1}}{B}
 \le
   2^{\mu(i, |u|) - \gamma(|u|)}
\end{align}
by \eqref{eq:d(i,0)g_u}, \eqref{eq:d(i,j)g_u} and our choice of $B$.

We have (\ref{enum:kf:end}) with
  $\rho(x) = \sigma(x) \cdot (\gamma(x)+r(x)+s(k)+k+3)$, because
  \begin{equation}
\label{eq: rho}
   h_u(1) 
   = \frac{H_u(p(|u|), 2^{q(|u|)})}{B^{d_u(p(|u|))}} 
   = \frac{L(u)}{2^{\sigma(|u|) \cdot (\gamma(|u|)+r(|u|)+s(k)+k+3)}} 
   = 2^{-\rho(|u|)} L(u). 
  \end{equation}
\end{proof}

We used the exponential positioning function $d _u$ 
defined at \eqref{eq:positioning}, 
so that we have $d _u (i + 1) > k \cdot d _u (i)$ for the 
second inequality of \eqref{equation: upperbound of differences}
for $k \geq 2$. 
Because of this, we had to restrict $p$ to be logarithmic, 
because otherwise the function $\sigma$ in \eqref{eq:positioning}
would have to be superpolynomial and so would $\rho$ in \eqref{eq: rho}. 

The proof of Lemma~\ref{DifferentiableFamily} is analogous, 
starting with the $L$ and $(G_u)_u$ of Lemma~\ref{DIVPpolyIsPSPACEhard}. 
The only difference is that $p$ is now a polynomial 
and therefore we use $d _u (i) = i$ in \eqref{eq:positioning}
(and $\sigma = p$ in \eqref{eq: rho}). 

\subsection{Proof of the Main Theorems}
\label{subsection: proof of theorems}
Using the function families $(g_u)_u$ and $(h_u)_u$ 
obtained from Lemmas \ref{KTimesFamily} or \ref{DifferentiableFamily}, 
we construct the functions $g$ and $h$ in 
Theorems \ref{DifferentiableIsPspace} and \ref{KTimesIsCH} as follows. 
Divide $[0,1)$ into infinitely many subintervals $[l^-_u, l^+_u]$,
with midpoints $c_u$.
We construct $h$ by putting a scaled copy of $h_u$ onto $[l^-_u, c_u]$ and
putting a horizontally reversed scaled copy of $h_u$ onto $[c_u, l^+_u]$ 
so that $h(l^-_u) = 0$, $h(c_u) = 2^{-\rho'(|u|)} L(u)$ and $h(l^+_u) = 0$ where $\rho'$ is a polynomial.
In the same way, $g$ is constructed from $(g_u)_u$ so that $g$ and $h$ satisfy \eqref{eq:ode}.
We give the details of the proof of 
Theorem~\ref{KTimesIsCH} from Lemma~\ref{KTimesFamily}, 
and omit the analogous proof of Theorem~\ref{DifferentiableIsPspace} 
from Lemma~\ref{DifferentiableFamily}.

\begin{proof}[Proof of Theorem~\ref{KTimesIsCH}]
Let $L$ and $\mu$ be as Lemma~\ref{KTimesFamily}.
Define
$
  \lambda(x) = 2x + 2
$, $
  \gamma(x) = \mu(x, x) + x \lambda(x)
$
and for each $u$ let
$
 \Lambda_u = 2^{\lambda(|u|)}
$, $
 c_u = 1-{2^{-|u|}}+{2\bar{u}+1}/{\Lambda_u}
$, $
 l_u^\mp = c_u\mp{1}/{\Lambda_u}
$,
 where $\bar u \in \{0, \dots, 2^{|u|} - 1\}$ is the number represented by $u$ in binary notation.
Let $\rho$, $(g_u)_u$, $(h_u)_u$ be as in Lemma~\ref{KTimesFamily} 
corresponding to the above $\gamma$.

We define
 \begin{align} \label{eq:g}
 g \left(l^\mp_u \pm \frac{t}{\Lambda_u}, \frac{y}{\Lambda_u}\right)
  &= \begin{cases}
      \pm \displaystyle \sum_{l=0}^k \frac{\D^{(0,l)}g_u(t,1)}{l!} (y-1)^l 
      &  \text{if} \ 1<y, \\
      \pm g_u(t, y)      & \text{if} \ {-1} \le y \le 1, \\
      \pm \displaystyle \sum_{l=0}^k \frac{\D^{(0,l)}g_u(t,-1)}{l!} (y+1)^l  
      &  \text{if} \ 1<y, \\
    \end{cases} 
  \\
 h \left( l^\mp_u \pm \frac{t}{\Lambda_u} \right) 
  & = \frac{h_u(t)}{\Lambda_u}
\end{align}
for each string $u$ and $t \in [0,1)$, $y \in [-1, 1]$.
Let $g(1,y) = 0$ and $h(1) = 0$ for any $y \in [-1,1]$.

It can be shown similarly to the Lipschitz version 
\cite[Theorem 3.2]{kawamura2010lipschitz}
that $g$ and $h$ satisfy \eqref{eq:ode} and $g$ is polynomial-time computable.
Here we only prove that $g$ is of class $\classC^{(\infty, k)}$.
We claim that 
for each $i \in \N$ and $j \in \{0, \dots, k\}$, 
the derivative $\D _1 ^i \D _2 ^j g$ is given by 
\begin{equation}
   \D_1^i \D_2^j g \left(l^\mp_u \pm \frac{t}{\Lambda_u}, \frac{y}{\Lambda_u}\right)
   = \begin{cases}
      \pm \Lambda_u^{i+j} \sum^{k}_{l=j} \frac{\D^{(i,l)} g_u(t,1)}{(l-j)!}
      (y - 1)^l &  \text{if } y < -1,
      \\
      \pm \Lambda_u^{i+j} \D^{(i, j)} g_u(t, y) & \text{if } {-1} \le y \le 1,
      \\
      \pm \Lambda_u^{i+j} \sum^{k}_{l=j} 
      \frac{\D^{(i,l)} g_u(t, -1)}{(l-j)!} (y + 1)^l &  \text{if } 1<y
    \end{cases}  \label{eq:d1id2jg}
\end{equation}
for each $l_u^\mp \pm t/\Lambda_u \in [0,1)$ and $y/\Lambda_u \in [-1, 1]$, 
and by $\D _1 ^i \D _2 ^j g (1, y) = 0$. 
This is verified by induction on $i + j$. 
The equation \eqref{eq:d1id2jg} follows from calculation 
(note that this means verifying 
that \eqref{eq:d1id2jg} follows from the definition of $g$ when $i = j = 0$; 
from the induction hypothesis about $\D _2 ^{j - 1} g$ when $i = 0$ and $j > 0$; 
and from the induction hypothesis about $\D _1 ^{i - 1} \D _2 ^j g$ when $i > 0$).
That $\D _1 ^i \D _2 ^j g (1, y) = 0$ is 
immediate from the induction hypothesis if $i = 0$. 
If $i > 0$, the derivative
$\D_1^i \D_2^j g (1, y)$ is by definition the limit 
\begin{equation}
\lim_{s \to 1 - 0} \frac{\D_1^{i-1} \D_2^j g(1, y) - \D_1^{i-1} \D_2^j g (s, y)}{1 - s}.
\label{eq:limitofderivative}
\end{equation}
This can be shown to exist and equal $0$, 
by observing that the first term in the numerator is $0$
and the second term is bounded, when $s \in [l ^- _u, l ^+ _u]$, by 
 \begin{align}
&
  \lvert
   \D_1^{i-1} \D_2^j g (s, y)
  \rvert
  \le 
  \Lambda_u^{i-1+j} \sum^{k}_{l=j} \lvert \D^{(i-1,l)} g_u \rvert \cdot (\Lambda_u + 1)^l 
  \notag
\\
& \qquad
 \le
  \Lambda_u^{i-1+j}  \cdot k \cdot 2^{\mu(i-1, |u|) - \gamma(|u|)} \cdot (2\Lambda_u)^k
  \notag
\\
& \qquad
  \le 2^{(i-1+j+k)\lambda(|u|) + 2k + \mu(i-1, |u|)  - \gamma(|u|)}
  \le 2^{-2 \lvert u \rvert}
  \le 2^{-\lvert u \rvert + 1} (1 - s), 
  \label{eq:sizeofderivative}
 \end{align}
where the second inequality is from 
Lemma~\ref{KTimesFamily} (\ref{enum:smooth})
and the fourth inequality holds for sufficiently large $\lvert u \rvert$
by our choice of $\gamma$. 
The continuity of $\D _1 ^i \D _2 ^j g$ on $[0,1) \times [-1, 1]$ follows
from \eqref{eq:d1id2jg} and Lemma~\ref{KTimesFamily} (\ref{enum:boundary}).
The continuity on $\{1\} \times [-1, 1]$ is verified by 
estimating $\D_1^{i} \D_2^{j} g$ similarly to \eqref{eq:sizeofderivative}. 
\end{proof} 

\section{Other Versions of the Problem}

\subsection{Complexity of the final value}
\label{section: final value}
In the previous section, 
we considered the complexity of the solution~$h$ of the ODE as a real function. 
Here we discuss the complexity of the final value $h (1)$ and
relate it to tally languages (subsets of $\{0\}^*$), 
as did 
Ko~\cite{ko1983computational} and 
Kawamura~\cite[Theorem~5.1]{kawamura2010complexity}
for the Lipschitz continuous case.

We say that a language~$L$ \emph{reduces to} a real number $x$ 
if there is a polynomial-time oracle Turing machine $M$ 
such that $M^\phi$ accepts $L$ for any name $\phi$ of $x$.
Note that this is the same as 
the reduction in Definition~\ref{definition: reduction}
to a constant function with value~$x$.

\begin{theorem}
\label{theorem: final value of once}
For any tally language $T \in \classPSPACE$,
there are a polynomial-time computable function
$g \colon [0,1] \times [-1,1] \to \R$ 
of class $\classC ^{(\infty, 1)}$ and 
a function $h \colon [0,1] \to \R$
satisfying \eqref{eq:ode} 
such that $L$ reduces to $h(1)$.
\end{theorem}

\begin{theorem}
\label{theorem: final value of fixed}
Let $k$ be a positive integer. 
For any tally language $T \in \classCH$,
there are a polynomial-time computable function
$g \colon [0,1] \times [-1,1] \to \R$ 
of class $\classC ^{(\infty, k)}$ and 
a function $h \colon [0,1] \to \R$
satisfying \eqref{eq:ode} 
such that $L$ reduces to $h(1)$.
\end{theorem}

We can prove Theorem~\ref{theorem: final value of fixed} 
from Lemma~\ref{KTimesFamily}
in the same way as the proof of \cite[Theorem~5.1]{kawamura2010complexity}.
We skip the proof of Theorem~\ref{theorem: final value of once}
since it is similar.

\begin{proof}
Let $T$ be any tally language in $\classCH$ and $k$ be any positive integer,
and let $L$ and $\mu$ be as Lemma~\ref{KTimesFamily}.
Define $
\lambda(x) = x + 1
$, $
\gamma(x) = \mu(x, x) + x \lambda(x)
$ and let $\rho$, $(g_u)_u$, $(h_u)_u$ be  as in Lemma~\ref{KTimesFamily} 
corresponding to the $\gamma$.
Since $L$ is $\classCH$-hard,
there are a polynomial-time function $F$ such that 
$T (0 ^i) =  L(F(0 ^i))$ for all $i$. 

Let $
l_n = 1 - 2^n
$ and $
\bar{\rho}(n) = \sum^{n-1}_{i = 0} \rho(|F(0 ^i)|)
$.  Define $g$ and $h$ as follows: 
when the first variable is in $[0,1)$, let
\begin{align}
 g \left(l_n + \frac{t}{2^{n+1}}, \frac{2m+(-1)^m y}{2^{2n+\gamma(n)+\bar{\rho}(n)}} \right)
 &=
 \frac{g_{F(0^n)}(t, y)}{2^{n-1+\gamma(n)+\bar{\rho}(n)}},
 \\
 h \left( l_n + \frac{t}{2^{n+1}} \right)
 &=
 \frac{h_{F(0^n)}(t)}{2^{2n+\gamma(n)+\bar{\rho}(n)}}
 + \sum^{n-1}_{i = 0} \frac{T (0^i)}{2^{2 i + \gamma (i) + \bar{\rho} (i + 1)}}
\end{align}
for each $n \in \N$, $t \in [0,1]$, $y \in [-1, 1]$ and $m \in \Z$; 
when the first variable is $1$,
let 
\begin{align} 
  g(1, y) 
&
 =
  0, 
\\
\label{equation: final value of h}
  h(1) 
&
 = 
 \sum^\infty_{n = 0} \frac{T (0^n)}{2^{2n+\gamma(n)+\bar{\rho}(n+1)}}. 
\end{align}
It can be proved similarly to the proof of Theorem~\ref{KTimesIsCH} 
that $g$ is polynomial-time computable and of class $\classC ^{(\infty, k)}$
and that $g$ and $h$ satisfy \eqref{eq:ode}.
The equation \eqref{equation: final value of h} implies 
that $T$ reduces to $h(1)$. 
\end{proof}

\subsection{Complexity of operators}
\label{section: constructive}

Both Theorems \ref{DifferentiableIsPspace} and \ref{KTimesIsCH}
state the complexity of the solution $h$ under the assumption 
that $g$ is polynomial-time computable.
But how hard is it to ``solve'' differential equations,
i.e., how complex is the operator that takes $g$ to $h$? 
To make this question precise,
we need to define the complexity of operators 
taking real functions to real functions.

Recall that, in order to discuss complexity of real functions,
we used string functions as names of elements in $\R$. 
Such an encoding is called a \emph{representation} of $\R$.
Likewise, 
we now want to encode real functions by string functions
to discuss complexity of real operators. 
In other words, we need to define representations of
the class $\classC _{[0, 1]}$ of continuous functions $h \colon [0,1] \to \R$ 
and class $\classLip _{[0, 1] \times [-1, 1]}$ of Lipschitz continuous functions $g \colon [0, 1] \times [-1, 1] \to \R$. 
The notions of computability and complexity depend on these representations.
Following \cite{kawamura2010operators},
we use $\deltabox$ and $\deltaboxLip$ as the 
representations of $\classC_{[0,1]}$ and $\classLip_{[0, 1] \times [-1, 1]}$, 
respectively.
It is known that 
$\deltabox$ is the canonical representation of $\classC_{[0, 1]}$ 
in a certain sense \cite{kawamura11:_funct_space_repres_and_polyn_time_comput}, 
and $\deltaboxLip$ is the representation defined by adding to $\deltabox$
the information on the Lipschitz constant.

Since these representations use string functions 
whose values have variable lengths,
we use \emph{second order polynomials}
to bound the amount of resources (time and space) of machines
\cite{kawamura2010operators}, 
and this leads to the definitions of second-order complexity classes
(e.g. $\classFPSPACEtwo$, polynomial-space computable),
reductions (e.g. $\redW$, polynomial-time Weihrauch reduction), 
and hardness.
Combining them with the representations of real functions mentioned above,
we can restate our theorems in the constructive form as follows.

Let $\OpIVP$ be the operator 
mapping a real function $g \in \classLip_{[0, 1] \times [-1, 1]}$ to
the solution $h \in \classC_{[0, 1]}$ of \eqref{eq:ode}.
The operator $\OpIVP$ is a partial function 
from $\classLip _{[0, 1] \times [-1, 1]}$ to $\classC _{[0, 1]}$
(it is partial because the trajectory may fall out of the interval $[-1, 1]$, 
see the paragraph following Theorem~\ref{KTimesIsCH}).
In \cite[Theorem 4.9]{kawamura2010operators}, the
$(\deltaboxLip, \deltabox)$-$\classFPSPACEtwo$-$\redW$-completeness of $\OpIVP$ 
was proved
by modifying
the proof of the results in the third row of Table~\ref{table:related}.
Theorem~\ref{DifferentiableIsPspace} can be rewritten in a similar way. 
That is, let $\OpIVP \mathord\upharpoonright ^{\classC ^{(\infty, 1)}}$ be the operator $\OpIVP$ 
restricted to class $\classC^{(\infty, 1)}$. Then we have: 

\begin{theorem}
\label{theorem: C1 constructive}
The operator $\OpIVP \mathord\upharpoonright ^{\classC ^{(\infty, 1)}}$ is $(\deltaboxLip, \deltabox)$-$\classFPSPACEtwo$-$\redW$-complete.
\end{theorem}

To show this theorem,
we need to verify that the information used to construct functions in the proof of Theorem~\ref{DifferentiableIsPspace}
can be obtained easily from the inputs.
We omit the proof since it does not require any new technique.
Theorem~\ref{theorem: C1 constructive}
automatically implies Theorem~\ref{DifferentiableIsPspace} 
because of \cite[Lemmas 4.7 and 4.8]{kawamura2010operators} and 
the $\redW$ versions of \cite[Lemmas 3.11 and 3.12]{kawamura2010operators}. 

In contrast, 
the polynomial-time computability in the analytic case
(the last row of Table~\ref{table:related})
is \emph{not} a consequence of a statement based on $\deltabox$. 
It is based on the calculation of the Taylor coefficients, 
and hence we only know how 
to convert the Taylor sequence of $g$ at a point to that of $h$. 
In other words, 
the operator $\OpIVP$ restricted to the analytic functions
is not $(\deltaboxLip, \deltabox)$-$\classFPtwo$-computable, 
but $(\deltaTaylor, \deltaTaylor)$-$\classFPtwo$-computable, 
where $\deltaTaylor$ is the representation that 
encodes an analytic function using its Taylor coefficients at a point 
in a suitable way. 
More discussion on representations of analytic functions 
and the complexity of operators on them 
can be found in 
\cite{kawamura12:_unifor_polyt_comput_operat_univar}.

\end{document}